\documentclass[conference,letterpaper]{IEEEtran}
\usepackage{amssymb,stmaryrd,amsmath,amsfonts,rotating}%
\usepackage{amsmath,amsthm,amssymb,cases}
\usepackage{algorithm}
\usepackage{algpseudocode}
\usepackage{amsmath,cases}
\usepackage{bbm}
\usepackage{amsthm}
\usepackage{color}
\usepackage{amssymb}
\usepackage{bm}
\usepackage[T1]{fontenc}
\usepackage{lmodern} %

\theoremstyle{definition}%

\newtheorem{teiri}{Theorem}

\usepackage{mathtools}
\mathtoolsset{showonlyrefs}

\usepackage{bbm}
\usepackage{mathrsfs}

\makeatletter
\newcommand\RedeclareMathOperator{%
  \@ifstar{\def\rmo@s{m}\rmo@redeclare}{\def\rmo@s{o}\rmo@redeclare}%
}
\newcommand\rmo@redeclare[2]{%
  \begingroup \escapechar\m@ne\xdef\@gtempa{{\string#1}}\endgroup
  \expandafter\@ifundefined\@gtempa
     {\@latex@error{\noexpand#1undefined}\@ehc}%
     \relax
  \expandafter\rmo@declmathop\rmo@s{#1}{#2}}
\newcommand\rmo@declmathop[3]{%
  \DeclareRobustCommand{#2}{\qopname\newmcodes@#1{#3}}%
}
\@onlypreamble\RedeclareMathOperator
\makeatother

\definecolor{grey}{rgb}{0.7, 0.75, 0.71}

\def\dark-red#1{\textcolor[rgb]{0.7,0.0,0.0}{#1}}

\definecolor{amber}{rgb}{1.0, 0.75, 0.0}

\def\...{\dotsc}

\def\intT2{\int_{-T/2}^{T/2}}

\def\sumi1n{\sum_{i=1}^{n}}
\def\sumi1N{\sum_{i=1}^{N}}
\def\sumi0N--{\sum_{i=0}^{N-1}}

\def\ccc{\cdots}

\def\=def{\overset{\text{\small def}}{=}}

\DeclareMathOperator*{\rank}{rank}

\DeclarePairedDelimiterX{\inp}[2]{\langle}{\rangle}{#1, #2}

\def\Fb{\mathbb{F}}

\def\Cc{\mathcal{C}}

\def\CcT{\tilde{\mathcal{C}}}

\def\<{\langle}
\def\>{\rangle}

 \def\mat4#1#2#3#4{
\begin{pmatrix}
 #1&\ccc&#2\\
 \vdots&&\vdots\\
 #3&\ccc&#4
\end{pmatrix}}

\def\0sf{\mathsf{0}}
\def\1sf{\mathsf{1}}

\def\Erm{\mathrm{E}}

\def\0BS{\boldsymbol{0}}
\def\1BS{\boldsymbol{1}}

\def\0B{\mathbf{0}}
\def\1B{\mathbf{1}}

\def\0H{\hat{0}}
\def\1H{\hat{1}}

\def\xH{\hat{x}}

\def\zH{\hat{z}}

\def\EH{\hat{E}}

\def\HH{\hat{H}}

\def\+TT{\texttt{+}}
\def\-{\texttt{-}}
\def\+KB{|+\> \<+|}
\def\-KB{|-\> \<-|}

\def\q0{|0\>}

\def\xiH{\hat{\xi}}

\def\zeroU{\underline{0}}
\def\0U{\underline{0}}
\def\1U{\underline{1}}
\def\aU{\underline{a}}

\def\sU{\underline{s}}
\def\tU{\underline{t}}

\def\vU{\underline{v}}
\def\wU{\underline{w}}
\def\xU{\underline{x}}

\def\zU{\underline{z}}

\def\0UH{\underline{\0H}}
\def\1UH{\underline{\1H}}

\def\xUH{\underline{\xH}}

\def\zUH{\underline{\zH}}

\def\JO{\bar{J}}

\def\tauU{\underline{\tau}}

\def\sigmaU{\underline{\sigma}}

\def\zetaU{\underline{\zeta}}

\def\xiU{\underline{\xi}}

\def\sigmaUH{\hat{\sigmaU}}

\def\zetaUH{\hat{\underline{\zeta}}}

\def\xiUH{\hat{\xiU}}

 \def\Ct{\mathtt{C}}
 
\RedeclareMathOperator{\Im}{Im}

\def\rank{\mathrm{rank}}

\def\thepage{\arabic{page}}
\newcommand{\T}{\mathsf{T}}

\newcommand{\I}{\mathbbm{1}}

\ifCLASSINFOpdf
\else
\fi
\begin{document}
\title{Efficient Mitigation of Error Floors in Quantum Error Correction using Non-Binary Low-Density Parity-Check Codes}
\author{
\IEEEauthorblockN{Kenta Kasai}
\IEEEauthorblockA{
Institution of Science Tokyo\\
Email: kenta@ict.eng.ist.ac.jp}
}
\maketitle

\begin{abstract}
In this paper, we propose an efficient method to reduce error floors in quantum error correction using non-binary low-density parity-check (LDPC) codes. We identify and classify cycle structures in the parity-check matrix where estimated noise becomes trapped, and develop tailored decoding methods for each cycle type. 
For Type-I cycles, we propose a method to make the difference between estimated and true noise degenerate. 
Type-II cycles are shown to be uncorrectable, while for Type-III cycles, we utilize the fact that cycles in non-binary LDPC codes do not necessarily correspond to codewords, allowing us to estimate the true noise. 
Our method significantly improves decoding performance and reduces error floors.
\end{abstract}
\IEEEpeerreviewmaketitle

\section{Introduction}
Recent progress in quantum computing has enabled the development of systems comprising tens of reliable logical qubits, built from thousands of noisy physical qubits~\cite{bluvstein2024logical}. Nevertheless, many important applications of quantum computing demand computations involving thousands or even more logical qubits~\cite{preskill2018quantum}. This underscores the pressing need for highly efficient quantum error correction methods capable of supporting large-scale logical qubit architectures.

In \cite{komoto2024quantumerrorcorrectionnear}, a generalization of the construction method of \((J=2,L)\)-regular non-binary low-density parity-check (LDPC) codes from \cite{6017122} was proposed, where the sum-product (SP) algorithm was applied to simultaneously decode X and Z errors in quantum codes. The method achieves performance close to the hashing bound while maintaining scalability and ensuring a constant coding rate as the code length increases. 
In~\cite{komoto2025explicitconstructionclassicalquantum}, a deterministic construction achieving girth 12 was proposed, along with an extension to spatially coupled codes~\cite{CSS_SC_ISIT}. Furthermore,~\cite{kasai2025quantumerrorcorrectiongirth16} presented a randomized construction method achieving girth 16.

For small values of \(L\), the existence of cycles of length $2L$ in the parity-check matrices causes the SP algorithm to fail in estimating the noise exactly, resulting in a high error floor. 
In quantum codes, even if the estimated noise \(\EH\) does not exactly match the true noise \(E\), if \(\EH^\dagger E\) forms a degenerate error, the codeword can still be recovered. 
Using this property, this paper explores an efficient method for searching for \(\EH\) such that \(\EH^\dagger E\) forms a degenerate error, for the codes in \cite{komoto2024quantumerrorcorrectionnear} and \cite{6017122}.

In this paper, we propose a method to reduce the error floor by applying postprocessing with a computational complexity independent of the code length. 
First, we identify the cycle structures in the parity-check matrix where the estimated noise, which causes the error floor, becomes trapped. 
We then classify these cycles into three types and propose a decoding method based on this classification.
For Type-I cycles, we propose a method to determine the estimated noise such that the difference between the estimated noise and the true noise becomes degenerate. 
It is shown that Type-II cycles cannot be corrected. 
For Type-III cycles, we use the fact that in the case of non-binary codes, unlike the binary case, cycles in \((J=2,L)\)-regular codes do not necessarily correspond to codewords. 
Based on this, we propose a method to estimate the true noise.

\section{Preparation}

The Pauli group $\mathcal{P}_n$ is non-commutative; however, by ignoring the global phase factor $\alpha$ of Pauli operators, we obtain the quotient group $\mathcal{P}_n /\{\pm I, \pm i I\}$, which is the subgroup $\{\pm I, \pm i I\} \subset \mathcal{P}_n$ modulo. Here, $I$ represents the identity operator on the space $\mathbb{C}^{2^n}$. This quotient group is isomorphic to the commutative group $\mathbb{Z}_2^{2n}$, and the isomorphism for $E\in \mathcal{P}_n /\{\pm I, \pm i I\}$ is given as follows: 
\begin{align}
 &E=\alpha \bigotimes_{i=1}^n X^{x_i} Z^{z_i} \leftrightarrow (\xU|\zU)=\left(x_1, \ldots, x_n \mid z_1, \ldots, z_n\right).\label{003014_6Jan25}
\end{align}
Using this correspondence, we identify the Pauli error $E$ with its binary representation $(\xU|\zU)$. Whether the vector is a row vector or column vector should be understood from the context.

A very important subclass of stabilizer codes \cite{gottesmanthesis} is the CSS codes \cite{calderbank96,steane96}. CSS codes are a type of stabilizer code characterized by the property that the non-identity components of each stabilizer generator within a tensor product are all equal to $X$ or all equal to $Z$. 

The check matrix of the stabilizer group $S$ is  expressed as binary vectors arranged in rows. 
Without loss of generality, the check matrix $H$ of a CSS code of length $n$ can be expressed as follows:
\begin{align}
 H = 
\begin{pmatrix}
 H_X & O \\
 O & H_Z
 \end{pmatrix}\in \Fb_2^{(m_X+m_Z)\times 2n},\label{222417_12Jan25}
\end{align}
where $H_X, H_Z$ are binary matrices of sizes $m_X\times n$ and $m_Z\times n$, respectively. The commutativity of the stabilizer generators implies that $ H_X H_Z^\T = O$. For simplicity, this paper assumes $m=m_X=m_Z$. The condition $ H_X H_Z^\T = O$ is equivalent to the following:
\begin{align}
 C_Z^{\perp} \subset C_X, \quad C_X^{\perp} \subset C_Z,\label{013431_22Jan25}
\end{align}
where $C_X,C_Z$ are the code spaces defined by $H_X,H_Z$ when considered as parity-check matrices. 
Furthermore, the dimension $k$ of the CSS code of length $n$ can be determined by the following formula:
\[
k = n - \operatorname{rank} H_X - \operatorname{rank} H_Z.
\]

Let us consider the codeword state $|\psi\>$ with a Pauli error $E\leftrightarrow (\xU,\zU)$ applied, resulting in the state $E|\psi\>$. The decoder considered in this paper performs decoding according to the following steps:
\begin{enumerate}
 \item Measure the syndrome $(\sU=H_Z\xU, \tU=H_X\zU)\in \Fb_2^{2m}$.
 \item Based on the syndrome, estimate the noise as $\EH\leftrightarrow(\xUH,\zUH)$ that satisfies the following:
 \begin{align}
 H_X\zUH=\sU, \quad H_Z\xUH=\tU\label{030159_23Jan25}
 \end{align}
 \item 
The codeword $|\psi\>$ is affected by noise $E$, resulting in the state $E|\psi\>$. Applying $\EH^\dagger$ to this state yields $E^\dagger E|\psi\>$.
If $\EH^\dagger E\in S$, then $\EH^\dagger E|\psi\>\propto |\psi\>$, and the codeword is recovered. 
\end{enumerate}

The condition $\EH^\dagger E\in S$ is equivalent to the existence of $\aU=(a_1,\ldots,a_{2m})^\T \in\Fb_2^{2m}$ such that $\prod_{i=1}^m S_i^{a_i}=\EH^\dagger E$, where $S_i$ are the stabilizer generators for $i=1,\ldots,2m$. Moreover, since $\EH^\dagger E\leftrightarrow(\xU+\xUH|\zU+\zUH)$, this is equivalent to:
\begin{align}
\sum_{i=1}^m {a_i} \sU_i=(\xU+\xUH|\zU+\zUH), \quad \text{for } i=1,\ldots,m, 
\end{align}
where $S_i\leftrightarrow\sU_i$. Writing this in terms of matrices and vectors, we get $  H^\T\aU=(\xU+\xUH|\zU+\zUH)$. Specifically, for CSS codes, since the parity-check matrix is given by \eqref{222417_12Jan25}, we use $\aU=(\aU_X|\aU_Z)$ to obtain:
\begin{align}
  & (H_X)^\T \aU_X=\xU+\xUH, \quad (H_Z)^\T \aU_Z=\zU+\zUH\label{025725_23Jan25}
\end{align}
which is equivalent to each of \eqref{041406_21Jan25} and \eqref{021354_20Jan25}:
\begin{align}
&G_X\left(\xU+\xUH\right)=0, \quad G_Z\left(\zU+\zUH\right)=0,\label{041406_21Jan25}
\\& \xUH+\xU\in C_X^\perp,\quad \zUH+\zU\in  C_Z^\perp, \label{021354_20Jan25}
\end{align}
where $G_X,G_Z$ are generator matrices of $H_X,H_Z$, respectively. 

From \eqref{013431_22Jan25}, it follows that \eqref{021354_20Jan25} implies
\begin{align}
 H_X\left(\zU+\zUH\right)=0, \quad H_Z\left(\xU+\xUH\right)=0,\label{022112_22Jan25}
\end{align}
which is equivalent to \eqref{030159_23Jan25}. The following is  sufficient for \eqref{021354_20Jan25}: 
\begin{align}
\xU+\xUH=H_X^{(i)}, \quad \zU+\zUH=H_Z^{(i')}\label{022213_22Jan25}
\end{align}
for some $i,i', $where $H_Z^{(i)}, H_X^{(i')}$ are the $i$-th and $i'$-th row vectors of $H_Z$ and $H_X$, respectively. 

\section{Code Construction}
In this paper, the primary target code uses the parameters $J=2, L=6, q=2^e$, and $e=8$. Codes with varying values of the parameter $L$ are used as comparisons.
Let $N=LP$, $M=JP$, $n=eN$, and $m=eM$. The finite field of size $q$ is denoted as $\Fb_q$.
Using the methods proposed in \cite{6017122} and \cite{komoto2024quantumerrorcorrectionnear}, $H_X$ and $H_Z$ are constructed through the following outline steps. 
\begin{enumerate}
\item Generate a pair of orthogonal $\Fb_2$-valued $(J=2, L, P)$ protograph matrices $(\HH_X, \HH_Z) \in (\Fb_2^{M \times N})^2$ such that the girth is 12.
 \item Generate a pair of orthogonal $\Fb_q$-valued $(J=2, L, P)$ protograph matrices $(H_{\Gamma}=(\gamma_{ij}), H_{\Delta}=(\delta_{ij})) \in (\Fb_q^{M \times N})^2$, having nonzero elements at the same positions as $(\HH_X, \HH_Z)$.
 \item  Generate a pair of orthogonal $\Fb_2$ matrices $(H_X=(A(\gamma_{ij})), H_Z=(A^\T(\delta_{ij}))) \in (\Fb_2^{n \times m})^2$, having nonzero submatrices of size $e\times e$ at the same positions as the nonzero elements of $(\HH_X, \HH_Z)$.
\end{enumerate}
Note that $A(\cdot)$ and $A^\T(\cdot)$ are isomorphism maps from $\Fb_q$ to $\Fb_2^{e\times e}$, as defined in \cite[Appendix B]{komoto2024quantumerrorcorrectionnear}. 
The CSS code defined by the parity-check matrices \eqref{222417_12Jan25} constructed using this $(H_X, H_Z)$ is the code considered in this paper.

Let $G_\Gamma$ and $G_\Delta$ be the corresponding $\Fb_q$-valued generator matrices when $H_\Gamma$ and $H_\Delta$ are regarded as $\Fb_q$-valued parity-check matrices. 
The next theorem give us the finite field representation of the error-correcting condition $\EH^\dagger E \in S$, which is equivalent to \eqref{041406_21Jan25}, the necessary condition \eqref{022112_22Jan25}, and the sufficient condition \eqref{022213_22Jan25}. 
\begin{teiri}[Finite field representation of $\EH^\dagger E\in S$]\label{165121_21Jan25}
For $s_i, t_i, x_j, z_j \in \mathbb{F}_2^e$, we define $\sigma_i, \tau_i, \xi_j, \zeta_j \in \mathbb{F}_q$ are such that $\underline{w}\left(\sigma_i\right)=s_i, \underline{v}\left(\tau_i\right)=t_i, \underline{w}\left(\xi_j\right)=$ $x_j$ and $\underline{v}\left(\zeta_j\right)=z_j$ hold. 
Note that $\underline{v}(\cdot)$ and $\underline{w}(\cdot)$ were defined in \cite[Appendix B]{komoto2024quantumerrorcorrectionnear}.
The following statements hold: 
\eqref{041406_21Jan25}, \eqref{022112_22Jan25}, and \eqref{022213_22Jan25} are equivalent to \eqref{124833_23Jan25}, \eqref{124841_23Jan25}, and \eqref{124848_23Jan25}, respectively.
\begin{align}
 & G_\Gamma(\xiU+\xiUH)=\underline{0}, \hspace{-1cm} && G_\Delta(\zetaU+\zetaUH)=\underline{0}, \label{124833_23Jan25} \\
 & H_\Gamma\left(\zetaU+\zetaUH\right)=0,  \hspace{-1cm}&& H_\Delta\left(\xiU+\xiUH\right)=0, \label{124841_23Jan25} \\
 & \xiU+\xiUH=H_\Gamma^{(i)}, \hspace{-1cm} && \zetaU+\zetaUH=H_\Delta^{(i')}, \label{124848_23Jan25}
\end{align}
where $H_\Gamma^{(i)}$ and $H_\Delta^{(j)}$ are the $i$-th and $j$-th row vectors of $H_\Gamma$ and $H_\Delta$, respectively.
\end{teiri}
\begin{proof}
The proof is evident from the discussion in \cite[Appendix B]{komoto2024quantumerrorcorrectionnear}.
\end{proof}
\section{Conventional Noise Estimation Method}
In this section, we review the noise estimation method presented in \cite{komoto2024quantumerrorcorrectionnear}. The method involves dividing noise vectors into segments, modeling their probabilities under a depolarizing channel, and applying the SP iterative decoding algorithm to estimate noise. 

First, we divide the noise vectors $\underline{x}, \underline{z}\in \Fb_2^n$ into $e$-bit segments and write:
\[
\xU=\left(x_1, \ldots, x_N\right), \quad \zU=\left(z_1, \ldots, z_N\right),
\]
where $x_j, z_j\in \Fb_2^e$.  
For simplicity in notation, we omit random variables and write $\Pr(X=x)$ as $p(x)$.  
The probability of occurrence for \(\xU, \zU \in \Fb_2^{eN}\) under the depolarizing channel with probability $p_D$ is expressed as:
\begin{align}
    p(\xU,\zU) &= \prod_{j=1}^N p(x_{j},z_{j}), \\
    p(x_{j},z_{j}) &= \prod_{k=1}^e p(x_{j}^k,z_{j}^k), \\
    p(x,z) &=
    \begin{cases}
        1-p_D, &  (x,z)=(0,0), \\
        \frac{p_D}{3}, & (x,z)=(0,1),(1,0),(1,1),
    \end{cases}
\end{align}
where $x_j^k, z_j^k\in \Fb_2$ are the $k$-th bits of $x_j$ and $z_j$, respectively.  
The flip probabilities $f_m$ for X and Z errors are obtained by marginalizing $p(x,z)$ as $f_m = 2p_D/3$.

The syndromes $\sU, \tU$ are similarly divided into $e$-bit segments:
\[
\sU=(s_1,\ldots,s_N)=H_Z\xU, \quad \tU=(t_1,\ldots,t_N)=H_X\zU,
\]
where $s_i, t_i \in \Fb_2^e$.  
The posterior probability of the noise given the syndromes, $p(\xU,\zU|\sU,\tU)$, factorizes as shown in \eqref{052739_20Jan25}.  
It can also be expressed as a function in $\Fb_q$ as shown in \eqref{154500_21Jan25} \cite[Appendix C.2]{komoto2024quantumerrorcorrectionnear}.  
For $s_i, t_i, x_j, z_j \in \mathbb{F}_2^e$, we define $\sigma_i, \tau_i, \xi_j, \zeta_j \in \mathbb{F}_q$ such that $\underline{w}\left(\sigma_i\right)=s_i$, $\underline{v}\left(\tau_i\right)=t_i$, $\underline{w}\left(\xi_j\right)=x_j$, and $\underline{v}\left(\zeta_j\right)=z_j$.  
Note that $\underline{v}(\cdot)$ and $\underline{w}(\cdot)$ were defined in \cite[Appendix B]{komoto2024quantumerrorcorrectionnear}.  
The variables $H_\Gamma, H_\Delta, \xiU, \zetaU, \sigmaU, \tauU$ can be regarded as the $\Fb_q$-valued representations of $H_X, H_Z, \xU, \zU, \sU, \tU$, respectively.  
\begin{figure*}
 \begin{align}
    p(\xU,\zU|\sU,\tU)
 &\propto    \Bigl(\prod_{i=1}^{M} \I\bigl[\sum_{j=1}^{N}\textcolor{black}{(H_Z)}_{ij}x_j=s_i\bigr] \Bigr)
    \Bigl( \prod_{i=1}^M \I\bigl[\sum_{j=1}^{N}\textcolor{black}{(H_X)}_{ij}z_j=t_i\bigr] \Bigr)
     \prod_{j=1}^N p\left(x_j, z_j\right) \label{052739_20Jan25}
 \\&=    \Bigl( \prod_{i=1}^M \I\bigl[\sum_{j=1}^N\textcolor{black}{\delta_{ij}}\xi_j=\sigma_i\bigr] \Bigr)
    \Bigl( \prod_{i=1}^M \I\bigl[\sum_{j=1}^N\textcolor{black}{\gamma_{ij}}\zeta_j=\tau_i\bigr] \Bigr)
     \prod_{j=1}^N p\bigl(\wU(\xi_j), \vU(\zeta_j)\bigr) .   \label{154500_21Jan25}
 \end{align}
\end{figure*}

Using the iterative SP algorithm, the estimated noise $(\xiUH^{(\ell)}, \zetaUH^{(\ell)})$ is determined in each SP iteration round $\ell$.  
In \cite{komoto2024quantumerrorcorrectionnear}, the stopping condition and decoding success/failure determination were performed according to the procedure in Algorithm \ref{013401_23Jan25}.  
This criterion judged harmless degenerate errors as decoding failures, resulting in a strict criterion.

\begin{algorithm}[t]
\caption{Success/failure determination algorithm for estimated noise $(\xiUH,\zetaUH)$ for true noise $(\xiU,\zetaU)$ \cite{komoto2024quantumerrorcorrectionnear}.}\label{alg:decoding}
\label{013401_23Jan25}\begin{algorithmic}[1]
\State $\ell \gets 0$
\While{True}
    \State Update $(\xiUH^{(\ell)}, \zetaUH^{(\ell)})$ by one round of SP iteration.
    \If{$H_\Delta\xiUH^{(\ell)} = \sigmaU \land H_\Gamma\zetaUH^{(\ell)} = \tauU$}
	   \State $(\xiUH,\zetaUH)\gets (\xiUH^{(\ell)},\zetaUH^{(\ell)})$
        \If{$\xiUH=\xiU\land  \zetaUH= \zetaU$}
            \State \textbf{Output:}  Decoding Success
        \Else
            \State \textbf{Output:} Decoding Failure
        \EndIf
    \Else
        \State $\ell \gets \ell + 1$
    \EndIf
\EndWhile
\end{algorithmic}
\end{algorithm}

The frame error rate (FER) based on this criterion is plotted as the solid curves in Fig.~\ref{fig:rainbow}.  
Each curve group corresponds to $L=16, 10, 8, 6$, from left to right.  
For $L$ values other than $L=6$, the curves exhibit no high error floor, and a sharp threshold phenomenon is observed as the code length increases.  
The curve for $L=6$, however, has a high error floor, which becomes even higher as the code length increases.

In the next section, we provide the characteristics of noise estimation errors that cause the error floor in the case of codes with $L=6$.
\section{Classification of Noise Estimation Errors Causing the Error Floor}

This section focuses on situations in the error floor region where the estimation of X-noise fails.  
The estimated X-noise in the SP algorithm at iteration $\ell$ is denoted by $\xiUH^{(\ell)}$.  
The estimation of Z-noise $\zetaUH^{(\ell)}$ can be treated symmetrically in the same manner.  

Upon examining errors in the error floor region, it was observed that estimation errors $\xiUH^{(\ell)}(\neq \xiU^{(\ell)})$ for large $\ell$ 
were caused by estimated noise $\xiUH^{(\ell)}$ such that the Hamming distance between $\xiUH^{(\ell)}$ and $\xiU^{(\ell)}$ was $\le L$.  
This paper focuses exclusively on such estimation errors.
In this section, we further elucidate the graph-theoretical characteristics of the estimation errors.
Specifically, all the esitimation errors were observed to occur due to being trapped in a cycle of length $2L$.

To clearly describe the graph-theoretical characteristics, the Tanner graph corresponding to the matrix $H_\Delta$ is considered, with its columns treated as variable nodes, its rows as check nodes, and its submatrices treated as subgraphs interchangeably.
The set of all cycles of length $2L$ in $H_\Delta$ is denoted by $\Cc_{\Delta}^{(2L)}$.  

The set of indices of incorrectly estimated noise $\xiUH^{(\ell)}$ is denoted as $J_\Erm(\xiUH^{(\ell)})= \{ j \in [N] \mid \xiH^{(\ell)}_j \neq \xi_j \}$.
In most cases of decoding failure in the error floor region, the estimated noise $\xiUH^{(\ell)}$ was trapped in a cycle of length $2L$. Specifically,  
$J_\Erm(\xiUH^{(\ell)}) \subset J(\Ct_\Delta)  $
for some $\Ct_\Delta \in \Cc_{\Delta}^{(2L)}$.  

Such cycles $\Ct_\Delta$ can be classified into three types based on their characteristics.  
For this classification, some preliminary definitions are necessary.
The set of column indices corresponding to the non-zero entries of the $i$-th row $H_\Gamma^{(i)}$ of $H_\Gamma$ is denoted as $J_\Gamma^{(i)}$, where $i \in [M]$ and $[M] = \{1, \ldots, M\}$.  
The Tanner graph restricted to the columns of $H_\Delta$ indexed by $J_\Gamma^{(i)}$ is denoted by $\Ct_\Delta^{(i)}$.  
This graph, $\Ct_\Delta^{(i)}$, represents a single cycle of length $2L$, with the corresponding matrix having both row weight and column weight equal to 2.
Moreover, $\Ct_\Delta^{(i)}$ corresponds to an $\Fb_q$-valued matrix of size $L \times L$ and rank $L-1$ \cite{6017122,komoto2024quantumerrorcorrectionnear}.
The set of such $\Ct_\Delta^{(i)}$ for $i \in [M]$ is denoted by $\CcT_{\Delta}^{(2L)}$, and it holds that $\CcT_{\Delta}^{(2L)} \subset \Cc_{\Delta}^{(2L)}$.  

We classify the cycles $\Ct_\Delta \in \Cc_{\Delta}^{(2L)}$  into three types based on their membership in $\CcT_{\Delta}^{(2L)}$ and their rank, as follows:
\[
\begin{array}{lll}
 \text{Type I:} & \Ct_\Delta \in \CcT_{\Delta}^{(2L)}, & \\
 \text{Type II:} & \Ct_\Delta \notin \CcT_{\Delta}^{(2L)}, & \rank(\Ct_\Delta) < L, \\
 \text{Type III:} & \Ct_\Delta \notin \CcT_{\Delta}^{(2L)}, & \rank(\Ct_\Delta) = L.  
\end{array}
\]
In the following sections, we address the identification of the cycle in which estimation errors become trapped and the methods for correcting these estimation errors.
\section{Cycle Identification Method}

In this section, we propose a method to identify the cycle $\Ct_\Delta \in \Cc_{\Delta}^{(2L)}$ in which the estimated noise is trapped, i.e., satisfying $J_\Erm(\xiUH^{(\ell)}) \subset J(\Ct_\Delta)$.  
It should be noted that the decoder cannot directly know the error set $J_\Erm(\xiUH^{(\ell)})$.

Consider a sufficiently large iteration round $\ell$ where the SP decoding has sufficiently converged, or at most $L$ symbols are oscillating. 
For such $\ell$, let $J(\ell)$ denote the set of indices $j$ where $\xiH_j^{(\ell)}$ has changed during the past $L$ iterations:
\begin{align}
 J(\ell) = \{ j \mid \xiH_j^{(\ell'-1)} \neq \xiH_j^{(\ell')} \text{ for some } \ell' \in [\ell-L, \ell] \}.
\end{align}
Experiments revealed that when the size of $J(\ell)$ is not zero, it is always greater than 1.
Let $I(\ell)$ be the set of row indices corresponding to syndromes computed from the estimated noise $\xiUH^{(\ell)}$ that do not match the actual syndromes:
\begin{align}
 I(\ell) = \{ i \mid \sigma_i^{(\ell)} \neq \sigma_i, \sigmaUH^{(\ell)} = H_\Delta \xiUH^{(\ell)} \}.
\end{align}

We define $\hat{\Ct}_\Delta$ as the estimated cycle of $\Ct_\Delta$ that satisfies the following conditions:
\begin{align}
 J(\ell) \subset J(\hat{\Ct}_\Delta), \quad I(\ell) \subset I(\hat{\Ct}_\Delta).
\end{align}

Recall that the matrix $H_\Delta$ is constructed to avoid cycles of length 4.  
Moreover, since the column weight is $J=2$, no two cycles in $\Cc_{\Delta}^{(2L)}$ share the same pair of columns.  
Therefore, $\hat{\Ct}_\Delta$ can be identified from $J(\ell)$ with a computational complexity of $O(L)$. If identification fails, decoding failure is declared.

Next, it is determined whether $\hat{\Ct}_\Delta$ is Type-I (in which case $i$ is also identified) or Type-II/III.  
To do this, select any two distinct indices $j, j' \in J(\hat{\Ct}_\Delta)$ and identify the row $H_\Gamma^{(i)}$ that contains nonzero components corresponding to $j$ and $j'$.  
This can be achieved with a computational complexity of $O(L)$.  

If such a row $H_\Gamma^{(i)}$ exists, $\hat{\Ct}_\Delta$ is classified as Type-I; otherwise, it is a Type-II or Type-III cycle.  
The distinction between these two types is determined by calculating the rank of $\hat{\Ct}_\Delta$.
\section{Proposed Decoding Method}
In the previous section, the cycle $\hat{\Ct}_\Delta$ where $\xiUH^{(\ell)}$ was trapped was identified using its history over the past $L$ iterations.  
In this section, we propose a method to determine the new estimated noise $\xiUH$ using this cycle $\hat{\Ct}_\Delta$.  
Decoding success or failure is judged based on whether \eqref{124833_23Jan25} is satisfied. 

The decoder design assumes that the noise is correctly estimated outside the indices in $J:=J(\hat{\Ct}_\Delta)$, i.e., $\xiUH^{(\ell)}_{\JO} = \xiU_{\JO}$, wherer we define $\JO := [N] \setminus J$.  
  Thus, the task of the decoder is reduced to determining the estimated noise $\xiUH^{(\ell)}_J$ for the indices in $J$.  
The entire estimated noise $\xiUH$ is then constructed by concatenatin $\xiUH_{\JO}$ and $\xiUH_{J}$.  

From Theorem \ref{165121_21Jan25}, it is important to note that even if $\xiU_{J} \neq \xiUH^{(\ell)}_{J}$, decoding can still succeed as long as $G_\Gamma(\xiU + \xiUH) = \zeroU$ is satisfied, meaning the estimate does not harm the recovery of the codeword.
\subsection{Type-I}\label{212737_22Jan25}
To correct Type-I errors, decoding succeeds if an estimated error $\xiUH$ satisfying $\xiUH + \xiU = aH_\Gamma^{(i)}$ can be found, based on Theorem \ref{165121_21Jan25}.  
Let $J := J(\hat{\Ct}_\Delta)$ and $I := I(\hat{\Ct}_\Delta)$. The following steps determine such an estimated error $\xiUH$:
Restrict $\xiUH^{(\ell)}$ and $H_\Delta$ to the components in $\JO$, and denote them as $\xiUH^{(\ell)}_{\JO}$ and $(H_\Delta)_{\JO}$, respectively.  
The resulting syndrome, ignoring contributions from the components in $J$, is written as:
\[
(H_\Delta)_{\JO} \xiUH^{(\ell)}_{\JO} = \sigmaU_0.
\]
If there are no estimation errors outside $J$, i.e., $J_\Erm(\xiUH^{(\ell)}) \subset J$, the syndrome will match that of the true noise:
\[
(H_\Delta)_{\JO} \xiU_{\JO} = \sigmaU_0.
\]

Next, recall that $\hat{\Ct}_\Delta$ is a $\Fb_q$-valued matrix  of size $L \times L$ and  rank $L-1$.  
Solve the linear equation $\hat{\Ct}_\Delta \xiUH_J = (\sigmaU_0)_I$ to find a solution $\xiUH_J \in \Fb_q^L$.  
Since $\hat{\Ct}_\Delta$ forms a single cycle, this can be achieved with computational complexity $O(L)$.  
For indices outside $J$, assign $\xiUH_{\JO} := \xiUH^{(\ell)}_{\JO}$.  
Concatenating $\xiUH_J$ and $\xiUH_{\JO}$, construct $\xiUH$ as the estimated noise.

Suppose the solution $\xiUH_J$ chosen from the equation above differs from the true noise, $\xiUH_J \neq \xiU_J$.  
This does not affect the recovery of the codeword. 
To explain this, note that the true noise $\xiU_J$ is a particular solution to the equation $\hat{\Ct}_\Delta \xiUH_J = (\sigmaU_0)_I$, while the homogeneous solution is a scalar multiple of $(H_\Gamma^{(i)})_J$.  
Thus, $\xiUH_J$ can be expressed as:
\[
\xiUH_J = a(H_\Gamma^{(i)})_J + \xiU_J \quad \text{for any } a \in \Fb_q.
\]
Therefore, $\xiUH + \xiU = aH_\Gamma^{(i)}$, which satisfies $G_\Gamma(\xiUH + \xiU) = \zeroU$ by Theorem \ref{165121_21Jan25}, ensuring successful decoding.

\subsection{Type-II}
Consider the case where $\hat{\Ct}_\Delta$ is not full rank.  
In this case, experimental results showed that the SP decoder always converged.  
This implies that the incorrect estimated noise $\xiUH^{(\ell)}$ produces the same syndrome as the true noise $\xiU (\neq \xiUH^{(\ell)})$:  
\[
H_\Delta \xiU = H_\Delta \xiUH^{(\ell)}.
\]
As a result, $\xiUH^{(\ell)}$ is adopted as the estimated noise $\xiUH$, and decoding fails since $G_\Gamma(\xiU + \xiUH) \neq \zeroU$.  
Worse still, this failure cannot be detected by the decoder, leading to an undetected decoding failure.

To avoid such failures, codes can be designed to eliminate Type-II cycles.  
However, this paper does not explore this approach further.
\subsection{Type-III}
Finally, consider the case where $\hat{\Ct}_\Delta$ is full rank.  
In this case, $\xiUH_J$ is determined using the same approach as for Type-I.  
The difference from Type-I is that since $\hat{\Ct}_\Delta$ is full rank, the linear equation $\hat{\Ct}_\Delta \xiUH_J = (\sigma_0)_I$ has the true noise $\xiU_J$ as its unique solution.  
By solving this equation, the true noise $\xiU_J$ can be obtained. 
Using the method in \cite{5513602}, this can also be achieved with a computational complexity of $O(L)$.
\section{Numerical Experiments}
The decoding performance \((f_m, \mathrm{FER})\) of the proposed decoding method for the depolarizing channel, with parameters \(e=8\), \(J=2\), \(L=6\), \(P \in \{128, 1024, 8192\}\), and \(R=1/3\), is plotted in Fig.~\ref{fig:rainbow}. For comparison, conventional SP decoding for \(L \in \{6, 8, 10, 16\}\), \(P \in \{32, 128, 1024, 8192\}\), and \(R \in \{0.33, 0.50, 0.60, 0.75\}\) is also plotted as solid curves. The black solid line indicates the hashing bound. The code length is given by \(n = ePL\).

It is evident that conventional SP decoding alone, for \(L > 6\), simultaneously achieves both deep error floors and sharp thresholds. It is expected that the decoding failure probability can be reduced by increasing the code length for \(f_m\) smaller than these threshold points. The threshold points where the curves intersect are marked with circles. On the other hand, for \(L = 6\), a high error floor was observed. The error floor has been largely improved by the proposed method, but an error floor is still observed.

A closer look at the experimental results shows that all Type-I and Type-III estimation errors were corrected.
The remaining errors in the error floor consisted entirely of Type-II estimation errors.  
These errors are inherent to the code and cannot be addressed without redesigning the code construction. 
Conversely, if Type-II errors can be eliminated, it is possible to achieve a significantly deep error floor.
\begin{figure}[t]
    \centering
    \includegraphics[width=1.0\linewidth]{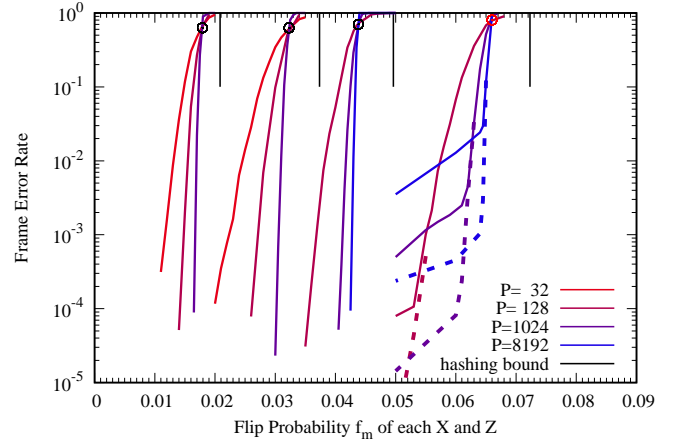}
    \caption{Decoding performance \((f_m, \mathrm{FER})\) of the proposed method (dashed).
    }
    \label{fig:rainbow}
\end{figure}

\section{Conclusion and Future Work}
In this paper, we proposed an efficient method for correcting estimation errors in quantum LDPC codes. 
Our method achieves a significant improvement in decoding performance, particularly by reducing the error floor, which has been a major challenge in previous approaches.

Nevertheless, some errors—especially those arising from Type-II estimation errors—remain an inherent challenge. 
These errors are closely tied to the structure of the code and cannot be easily resolved without reconsidering the code construction itself. 
Addressing these issues could enable even deeper reductions in the error floor, further enhancing the practical effectiveness of quantum error correction.

As part of future work, we plan to predict and mitigate the error floor by leveraging methods developed in the study of error floors for non-binary classical LDPC codes~\cite{5962725,6283931,e94-a_11_2144,e95-a_1_381,e95-a_12_2113}.

\bibliographystyle{IEEEtran}
\bibliography{IEEEabrv,00kasai} 
\end{document}